\documentclass[11pt,a4paper]{article}

\usepackage{graphicx}
\usepackage{amssymb,amsmath,amsthm}
\usepackage{color}
\usepackage{fullpage}
\usepackage{url}

\newcommand{\mb}{\mathbf} \newcommand{\mr}{\mathrm}
\newcommand{\be}{\begin{equation}}
\newcommand{\ee}{\end{equation}}
\newtheorem{lemma}{Lemma}
\newtheorem{theorem}{Theorem}

\title{Security of practical private randomness generation}
\author{Stefano Pironio and Serge Massar\\
Laboratoire d'Information Quantique, Universit\'e Libre de
Bruxelles, Brussels (Belgium)
}
\date{March 4, 2011}

\begin{document}
 \maketitle

\begin{abstract} 
Measurements on entangled quantum systems necessarily yield outcomes 
that are intrinsically unpredictable if they violate a Bell 
inequality. This property can be used to generate certified randomness in a device-independent way, i.e., without making detailed assumptions 
about the internal working of the quantum devices used to generate the random numbers. Furthermore these numbers are also private, i.e., they appear random not only to the user, but also to any adversary that might possess a perfect description of the devices.
Since this process requires a small initial random seed to sample the behaviour of the quantum devices and to extract uniform randomness from the raw outputs of the devices, one usually speaks of device-independent randomness \emph{expansion}.

The purpose of this paper is twofold. First, we point out that in most real, practical situations, where the concept of device-independence is used as a protection against unintentional flaws or failures of the quantum apparatuses, it is sufficient to show that the generated string is random with respect to an adversary that holds only classical-side information, i.e., proving randomness against quantum-side information is not necessary. Furthermore, the initial random seed does not need to be private with respect to the adversary, provided that it is  generated in  a way that is independent from the measured systems.  The devices, though, will generate cryptographically-secure randomness that cannot be predicted by the adversary and thus one can, given access to free public randomness, talk about private randomness \emph{generation}.
 
The theoretical tools to quantify the generated randomness according to these criteria were already introduced in \emph{S. Pironio et al, Nature \textbf{464}, 1021 (2010)}, but the final results were improperly formulated. The second aim of this paper is to correct this inaccurate formulation and therefore lay out a precise theoretical framework for practical device-independent randomness generation.
\end{abstract}

\section{Introduction}
Random numbers are essential for many applications such as computer simulations, statistical sampling, gambling, or video games. They are particularly important for classical and quantum cryptography, where the use of a flawed random number generator (RNG) can completely compromise the security. 
Many solutions have thus been proposed for the generation of random numbers (for recent work on random number generation see, e.g.,\cite{jennewein,stefanov,dynes,fiorentino,rng5,rng6,rng7}), but none is entirely satisfactory. As quoted from Wikipedia 
\cite{wikipedia}, every random number generator (RNG) is subject to the following problems:
\begin{quotation}
``It is very easy to misconstruct hardware or software devices which attempt to generate random numbers. Also, most `break' silently, often producing decreasingly random numbers as they degrade. A physical example might be the rapidly decreasing radioactivity of the smoke detectors \emph{[...]}.
Failure modes in such devices are plentiful and are complicated, slow, and hard to detect.

Because many entropy sources are often quite fragile, and fail silently, statistical tests on their output should be performed continuously. Many, but not all, such devices include some such tests into the software that reads the device.

Just as with other components of a cryptosystem, a software random number generator should be designed to resist certain attacks. Defending against these attacks is difficult. \emph{[...]}

\emph{[On estimating entropy]}. There are mathematical techniques for estimating the entropy of a sequence of symbols. None are so reliable that their estimates can be fully relied upon; there are always assumptions which may be very difficult to confirm. These are useful for determining if there is enough entropy in a seed pool, for example, but they cannot, in general, distinguish between a true random source and a pseudo-random generator.

\emph{[On performance test]}.
Hardware random number generators should be constantly monitored for proper operation. \emph{[...]}
Unfortunately, with currently available (and foreseen) tests, passing such tests is not enough to be sure the output sequences are random. A carefully chosen design, verification that the manufactured device implements that design and continuous physical security to insure against tampering may all be needed in addition to testing for high value uses."
\end{quotation}

Device-independent randomness generation aims to address these problems by exploiting the intrinsic unpredictability associated with the violation of Bell inequalities \cite{Bell,valentini,bhk}. More precisely, consider a quantum system composed of two separated parts $A$ and $B$ which upon receiving respective inputs $V^a$ and $V^b$, return respective outputs $X^a$ and $X^b$. If after $n$ successive uses of the devices, the \emph{observed} data violates a Bell inequality, it is then possible to \emph{certify} that the output string $(X^a_1,X^b_1),\ldots,(X^a_n,X^b_n)$ contains a certain amount of min-entropy, even when conditioned on the value of the inputs $(V^a_1,V^b_1),\ldots,(V^a_n,V^b_n)$, and a randomness extractor can therefore be applied to the outputs to obtain almost-uniform random bits. Furthermore, this conclusion can be reached independently of any detailed assumptions about the inner working of the devices and is thus immune to most of the problems mentioned above. 

That the violation of Bell inequalities is an indicator of quantum randomness had probably been recognized early on by many physicists, but was made explicit only recently in \cite{valentini, bhk, amg}. Not surprisingly, it was suggested shortly thereafter that Bell inequality violating systems could be exploited for randomness generation, and a scheme based on GHZ states was proposed in \cite{colbeck}. The possibility of device-independent randomness generation, however, was established only in \cite{dirng}, where a method to bound the min-entropy of the devices' output as a function of the observed Bell violation was introduced. Furthermore,  a	 proof-of-principle experimental demonstration was realised using two trapped ions.

The concept of device-independence (DI) is not restricted to randomness generation but includes adversarial applications such as quantum key distribution (QKD) \cite{ekert, Mayers 0,bhk,Acin et al. 4,Acin et al.} and coin tossing \cite{silman}, and non-adversarial ones such as state estimation \cite{Bardyn et al.}, entanglement witnesses \cite{BGLP}, and self-testing of quantum computers \cite{Magniez et al.}. In adversarial applications of device-independence it is often remarked that since the correctness of the protocol can be verified without making assumptions about the inner working of the devices, these could even have been prepared by the adversary itself. This has at least two implications as regards the theoretical analysis of device-independent randomness generation (and has also various implications for its experimental implementation, some of which will be briefly discussed later on).

First, if the adversary is allowed to prepare the quantum devices, nothing prevents him to entangle them with a quantum state that he keeps for himself in a quantum memory. It is then a priori possible that if he sees part of the devices' output at some later stage, he could measure his quantum state in a way that would give him useful information about the remainder of the output string. One thus needs to show that the output produced by the device  also appear random with respect to the quantum-side information held by the adversary. The methods introduced in \cite{dirng}, however, have been shown so far to estimate randomness only against classical-side information, i.e., against adversaries who do not share entanglement with the quantum devices.

Second, if the adversary happens to have some prior knowledge of the inputs used to sample the devices, he could exploit it to program the devices in a way that would mimick the violation of a Bell inequality while at the same time giving him substantial information about the generated outputs. A random, private seed is thus necessary to select the inputs and start off the protocol.  In addition, one also need some initial randomness to extract uniform random bits from the devices' outputs. One thus often speaks of device-independent randomness \emph{expansion} (DIRE). A scheme achieving quadratic expansion was presented in \cite{dirng}, where it was also suggested to use more than one pair of devices to obtain greater (e.g., exponential) expansion.

In this paper, we do not investigate this extremal adversarial scenario where the quantum devices have been acquired from a malicious provider. We are instead interested in the more real-life and practical situation where the manufacturer of the device is assumed to be honest, but where the concept of device-independence is used to provide an accurate estimation of the amount of randomness generated independently of noise, limited control of the apparatuses, or unintentional flaws of the devices. We point out in Section 2 that in this context it is sufficient to prove security against classical-information. Furthermore, the initial seed used to sample the devices and perform the randomness extraction does  not necessarily need to be private with respect to the adversary (it simply needs to be chosen in a way that is independent from the state of the devices). The output of the protocol,  though, will represent a private random string. In this case one can thus talk about private randomness \emph{generation}, given access to public randomness. (In the following, we will keep using the single terminology ``device-independent randomness expansion" to refer to the two situations in which the initial randomness is considered to be private or is viewed as a free, public resource).  

In Section 3, we then analyse the security of DIRE from this perspective. In particular, subsection 3.2. contains a detailed presentation of the model that we consider and of the assumptions on which it is based. Our main results are presented in subsection 3.3., where we show how to estimate the randomness produced in a Bell experiment if those assumptions are satisfied. Our analysis relies essentially on the tools introduced in \cite{dirng}, but importantly it fixes an issue that led to an improper formulation of the final results of \cite{dirng}\footnote{Specifically, the problem lies with Eq. (3) and Eq. (A.9) of the Supplementary Information of \cite{dirng} and with the final steps leading to these equations.}.  A very similar analysis has been presented in the independent work \cite{fgs}.  We briefly discuss how these results directly imply the security of various DIRE schemes in subsection 3.4.

Finally, we point out that a randomness-expansion scheme with superpolynomial expansion
and proven to be secure against quantum side information was recently introduced in \cite{vidick}. This protocol, however, requires an almost perfect violation of the CHSH inequality, while our results and those of \cite{fgs} are generic and holds for arbitrary Bell inequalities and any amount of violation\footnote{Note that previous versions of these results (see \cite{smv2} and \cite{fgsv2}) claimed security against quantum side information, but both proofs were incorrect.}.

\section{Honest vs dishonest device suppliers and DIRE}
The security of device-independent cryptographic protocols is based on a rather limited sets of assumptions, e.g., that the devices obey quantum theory, that separated devices can be prevented to communicate with one other, that the users of the device have access to a private source of randomness, and so on. Provided that these basic assumptions are satisfied, the security follows independently of implementation details such as the precise quantum states and measurement operators used, or the dimension of the Hilbert space in which they are defined. It is often stressed that security could thus in particular be guaranteed if the devices had been provided or sabotaged by the adversary itself. This possibility is fascinating from a conceptual point of view and deserves to be investigated for its own sake. However, it has probably little (if no) practical relevance, as has already been pointed out (see e.g., \cite{talk}). 

One reason is that while it is in principle possible to enforce the assumptions required for the security of a DI cryptography scheme based on malicious devices, in practice this may involve incredible technological and physical resources. For instance, how can we practically guarantee that the devices do not covertly leak out sensitive information to the adversary \cite{dirng}? How can we guarantee that they do not contain sneaky transmitters?  In principle communications through electromagnetic waves can be screened, but what about communications based on neutrinos or gravitational waves? When a ``door" is opened to let a particle enter in a device, how can we efficiently prevent other particles to come out of the device?

More generally, any practical cryptographic implementation, classical, quantum, or device-independent, will include and make use of classical computing and communicating devices to process, store, and transmit data. These classical devices, which are probably easier to corrupt than their quantum counter-parts, cannot be guaranteed secure if they have been acquired from dishonest providers. One should therefore either acquire these classical devices from trusted suppliers or inspect them for malicious behaviour. But then why apply a different standard to the quantum devices?

The real problem, to which the concept of device-independence offers a potential solution, is that even if the quantum devices have been obtained from honest suppliers or thoroughly inspected, many things can still unintentionally go wrong. Indeed, in standard (i.e., device-dependent) quantum cryptography, conclusions about the randomness or the secrecy of the outputs crucially depends on the \emph{physical} properties of the generation process, for example, on the fact that the outputs were produced by measuring the polarization of a single photon along well-defined directions. But then, how can one assess the level of security provided by a real-life implementation of a standard quantum cryptography protocol, which will inevitably differ in undetermined ways from the idealized, theoretical description \cite{scarani}? Consider for instance that the reported attacks \cite{makarov,lo,lo2} on commercial QKD systems did not exploit any intentional, maliciouss flaws in the devices.

This problem is particularly acute in the case of (classical or quantum) RNG devices, as it is very difficult even for honest parties to construct reliable RNGs and monitor them for proper operation. The generation of randomness in a device-independent way solves many of the shortcomings of usual RNGs listed earlier, since it makes possible an accurate estimation of the amount of randomness generated independently of noise, imperfections, lack of knowledge, or limited control of the apparatuses. 

The use of device-independence, even in a trusted provider situation, has the advantage over a full device-dependent approach that it requires only the verification of a limited number of precisely defined assumptions, on which the manufacturer of the device can focus. Furthermore, these assumptions can be much more easily enforced or verified with respect to the situation where the devices come from a dishonest provider, as one does not need to fight against devices that have been maliciously programmed\footnote{Note in particular that it is highly unlikely that the attack reported in \cite{bck} would spontaneously occur in non-malicious devices.}
\cite{dirng}. For instance, in the experiment reported in \cite{dirng} no particular measures have been taken to screen-off one device from the other. However, the experiments involve two atoms that are confined in two independent vacuum chambers separated by about 1 meter. At this distance, direct interaction between the atoms is negligible and classical microwave and optical fields used to perform measurements on one atom
have no influence on the other atom. Based on this superficial description of the setup, one can
safely assume that the two quantum systems are independent and that no imperfections, failures, or implementation weaknesses would lead to direct interaction between the devices (though imperfections could lead to other potential problems that can be ruled out by the DI approach), and thus that the general formalism used to derive a bound on the randomness applies. 

In the case of DIRE, assuming that the devices originate from a honest provider has not only experimental implications, but also theoretical ones. The first one is that, while the adversary may possess an arbitrarily accurate  classical description of the internal working of the devices at any given moment of time, it is highly unlikely that he could possess any quantum system that is entangled with those inside the devices if he did not manufacture or tamper with them. This means that proving that the output are random with respect to \emph{classical-side} information is sufficient. 

The second implication is that the adversary cannot program the devices to exploit any prior knowledge about the initial randomness used to choose the inputs. The inputs must still be selected in a way that is independent from the internal functioning of the devices, but this condition can be satisfied without having recourse to cryptographically-secure random number generators. For instance, in the experiment reported in \cite{dirng}, the measurement settings were chosen by combining through a XOR function several public random number generators that
use randomness derived from radioactive decay \cite{rand1}, atmospheric noise \cite{rand2}, and remote computer and network activity \cite{rand3}. While a dishonest manufacturer aware of this procedure could have exploited it in the design of the set-up, it is highly unlikely that the state of the ions in the experiment of \cite{dirng} was in any way correlated to the choice of measurement bases. If this condition is satisfied, it is justified, however, to conclude that the outputs of the devices do represent new, private random bits.  

Remark that the two above implications are specific to DIRE but would not hold for most DI cryptographic protocols. This is due to the fact that DIRE is a single user protocol completely carried out in a single secure lab and which therefore does not allow for the possibility of interactive attacks by the adversary. In contrast, DIQKD, for instance, usually involves the sending of quantum information between Alice and Bob's devices. This quantum information can be intercepted by the adversary and entangled with his own quantum system. Furthermore any knowledge of the random numbers used in the protocol could be exploited by the adversary to improve the efficiency of this interaction. Even if the devices are completely trusted, it is therefore still the case that the security of QKD must be based on a proof that holds against quantum-side information and that the random numbers used in the protocol must be cryptographically secure.

In the following section we analyse DIRE from the perspective discussed above, and show in particular how to prove the security of a DIRE protocol against classical side-information. 

\section{DIRE against classical side-information}
We start by recalling some definitions and results that will be
used in the following. We refer to \cite{dirng,shaltiel,extractors} for
more details. 

\subsection{Preliminaries}\label{prel}

\textbf{Random variables.} Let $R$ be a random variable over the finite set $\mathcal{R}$ and $\mr{Pr}[R=r]=P_R(r)$ the probability that it takes the value $r$. (In the following, we use upper-case letters to denote random variables and lower-case letters to denote specific values taken by these variables). The closeness between two distributions $P_R$ and $Q_r$ can be quantified through the trace distance
\be\label{td}
d(P_R,Q_r)=\frac{1}{2}\sum_{{r}}\mid P_{R}({r})-Q_R(r)|\,.
\ee
For simplicity, we will write $P(r)$ for the probabilities $P_R(r)$, when there is no risk of confusion.
Let $E$ be a random variable representing some classical side-information about the variable $R$, and let the correlations between $R$ and $E$ be described by a joint distribution $\mr{Pr}[R=r,E=e]=P_{RE}(re)$. We say that $R$ is $\delta$-random with respect to $E$ if it is $\delta$-close to a uniform distribution uncorrelated to $E$, that is if
\be \label{csecr}
d(P_{RE},U_R\times Q_E)=
\frac{1}{2}\sum_{{r,e}}\mid
P_{RE}({re})-U_R({r})\times Q_E({e})|\leq \delta \ee
for some distribution $Q_E$, where $U_R({r})=1/{|\mathcal{R}|}$ is
the uniform probability distribution on $\mathcal{R}$.

\textbf{Min-entropy.} The randomness of $R$ with respect to $E$ can be quantified through the conditional min-entropy 
\be
\label{aventr} H_{\mr{min}}(R|E)_P
=-\log_2\sum_{e\in\mathcal{E}} P_E(e)
\max_{r\in\mathcal{R}}\,P_{R|E}(r|e)\,. \ee The conditional
min-entropy (\ref{aventr}) is sometimes called the
\emph{average} conditional min-entropy to distinguish it from
the \emph{worst-case} conditional min-entropy defined by \be
\label{wsentr} \tilde H_{\mr{min}}(R|E)_P=-\log_2 \max_{r,e}
P_{R|E}(r|e)\,. \ee The worst-case min-entropy is a lower-bound
on the average min-entropy: $H_{\mr{min}}(R|E)_P\geq \tilde
H_{\mr{min}}(R|E)_P$. Note that when there is no
side-information $E$, both entropies reduce to the usual
definition $H_{\mr{min}}(R)_P=-\log_2 \max_{r\in\mathcal{R}}
P_R(r)$ for the classical min-entropy of a distribution $P_R$.

\textbf{Randomness extractors.} Given a $n$-bit string $R$ with a certain 
conditional min-entropy $k$ one can extract from it, using a randomness extractor and a small uniform seed $S$ , a new $m$-bit random string that is almost uniformly random. More formally, a function $\mr{Ext}:\{0,1\}^n\times\{0,1\}^d\rightarrow
\{0,1\}^m$ is a 
($m,k,\delta$)-strong extractor with uniform seed if for all distributions $P_{RE}$ with
$H_{\mr{min}}(R|E)_P\geq k$, and for a uniform seed $S\in\{0,1\}^d$, we have\footnote{Note that the definition of (classical) extractors does not usually involve side-information, but the definition given here and the conventional one can be shown to be essentially equivalent \cite{kt}.}
\be
\label{re}
d(P_{\mr{Ext}(R,S)SE},U_m\times P_S\times P_E)\leq \delta\,,
\ee
where $U_m$ is the uniform distribution on $\{0,1\}^m$. 
There exist different construction for randomness
extractors, characterized by different relations between the
parameters $n,m,d,k,\delta$. In particular, for any $k$ and $\delta$, there exist extractors with output length $m=k-4\log1/\delta-O(1)$ and seed length $d=O(\log^2(n/\delta)\log m)$ \cite{extractors}.

\textbf{Randomness and Bell experiments.} In \cite{dirng}, it was shown that there exists a fundamental, quantitative relation between the violation of Bell inequalities and the randomness produced in Bell experiments. We consider here for simplicity Bell experiments performed on two distinct systems $A$ and $B$, although our results generalize to more parties. We denote $V=(V^a,V^b)$ the measurement choices for systems $A$ and $B$ and assume that they each take values in a finite set $\mathcal{V}$. We denote the measurement outputs $X=(X^a,X^b)$ and assume that they each take values in the finite set $\mathcal{X}$. To any given input $V^a=v^a$, we can associate a set of measurement operators $\{M_A(x^a|v^a)\}_{x^a\in\mathcal{X}}$ such that $\sum_{x^a}M^\dagger_A(x^a|v^a)M_A(x^a|v^a)=I_A$, where $I_A$ is the identity operator on the Hilbert space $\mathcal{H}_A$ of system $A$. Similarly a set of measurement operators $M_B(x^b|v^b)$ can be associated to any given input $V^b=v^b$. The probability to obtain the pair of outputs $x=(x^a,x^b)$ given the pair of inputs $v=(v^a,v^b)$ when measuring a joint state $\rho_{AB}\in \mathcal{H}_A\otimes \mathcal{H}_B$ can then be written 
\be \label{qprob}
P(x|v)=\mr{tr}\left[M_{A}(x^a|v^a)\otimes M_B(x^b|v^b)\,\rho_{AB}\, M^\dagger_{A}(x^a|v^a)\otimes M^\dagger_B(x^b|v^b)\right]\,.
\ee 
A Bell expression $I$ is defined by a series of coefficients $c_{vx}$, which associate to a conditional probability distribution $P=\{P(x|v)\}$ the \emph{Bell expectation}
\begin{equation}
I[P]=\sum_{vx}c_{vx}P(x|v)\,.
\end{equation}
We denote by $I_q$ the maximal quantum Bell expectation, i.e., $I_q=\max_{P} I[P]$, where the maximum is taken over all distributions of the form (\ref{qprob}). 

In \cite{dirng} (see also \cite{nlrando}), it is shown that there exists a fundamental relation between the randomness of the distribution $P$ and the Bell expectation $I[P]$. More precisely, it is shown how using the semidefinite programming hierarchy introduced in \cite{sdp,sdp2}, one can compute for each $v$ a bound of the form 
\be \label{prI}
\max_x P(x|v)\leq g(I[P])\,,
\ee
which is valid for any state $\rho_{AB}$ and measurement operators $M_A(x^a|v^a)$, $M_B(x^b|v^b)$ such that (\ref{qprob}) holds. Here $g$ is a function that is concave (if not, we take its concave hull) and monotonically decreasing, taking values between $1$ and $1/|\mathcal{X}|^2$. In particular, it is thus also logarithmically concave. The above bound can be rewritten as $H_\mr{min}(X|V=v)_P\geq f(I[P])$ where $H_\mr{min}(X|V=v)_P=-\log_2 \max_x P(x|v)$ is the min-entropy of $X$ for given $v$, and $f(I[P])=-\log_2 g(I[P])$. From now on we refer to $g$ (or $f=-\log_2 g$) as a randomness bound associated to $I$.

\subsection{Modelling of the devices and basic assumptions}
We consider a single pair of Bell violating devices $A$ and $B$ (though the results below can be directly generalized to a multipartite setting), in which, the user Alice can respectively introduce inputs $V=(V^a,V^b)$ (the ``measurement settings") and obtains output $X=(X^a,X^b)$ (the ``measurement outcomes"). The quantum apparatuses are used $n$ times in succession for varying choices of the inputs. In full generality, the behaviour of the devices can be characterized by
\begin{itemize}
\item  an initial state $\rho_{AB}\in \mathcal{H}_A\otimes \mathcal{H}_B$;
\item a set $\mathcal{M}_{AB}=\{M_{AB}(x|v)\}$ of measurement operators on $\mathcal{H}_A\otimes \mathcal{H}_B$,  which have the product form 
\be \label{prod}
M_{AB}(x|v)=M_A(x^a|v^a)\otimes M_B(x^b|v^b)\,,
\ee
and which define the measurements applied on the state of the devices for given input $v=(v^a,v^b)$;
\item a joint unitary operation $U\in \mathcal{H}_A\otimes \mathcal{H}_B$, which is applied on the post-measurement state of the devices after each measurement and which represents the possibility for the devices to communicate between successive measurements (e.g., to establish new entanglement).
\end{itemize}
Note that to simplify the notation, we did not explicitly introduce a dependence of $M_{AB}(x|v)$ or $U$ on the measurement round $i$ or on the inputs and outputs obtained in previous steps, i.e., $M_{AB}(x|v)$ and $U$ are identical at each use of the devices. The above formulation is nevertheless completely general and can account for the possibility that the behaviour of the devices varies from one round to another and makes use of an internal memory. Indeed, the measurement operators $M_{AB}(x|v)$ and the operation $U$ can encode the value of the inputs $v$ and the output $x$ obtained in a given run in the post measurement state of the devices and ``read" back this information in the next step to perform an operation conditional on the previous history. The only restrictive hypothesis that we make is that the measurement operators have the product form (\ref{prod}).
Physically, this means that the systems $A$ and $B$ do not communicate with each other during the measurement itself.

We assume that the behaviour of the devices, characterized by the initial state $\rho_{AB}$, the set of measurement operators $\mathcal{M}_{AB}$, and the joint operation $U$, is perfectly known to the adversary. Note that the behaviour of the devices might depend on some external random parameters known or controlled by the adversary. For instance, the quality of the components used to produce the devices might vary in a way known to the adversary or he might control some parameters (such as temperature or changes in the voltage of the power supply) that can influence the output of the devices. This can be taken into account by assuming that the devices and the adversary's information are in a joint state
\be \label{csi}
\rho_{ABE}=\sum_e P(e) \rho_{AB}^e\otimes |e\rangle\langle e|\,,
\ee
where $\rho_{AB}=\sum_e P(e) \rho_{AB}^e$ and $e$ represents the knowledge that the adversary has on the state of the devices. We refer in the following to $(\rho_{ABE},\mathcal{M}_{AB},U_{AB})$ as the \emph{device behaviour}. Our assumption of classical side-information lies in the fact that the devices and the adversary are only classically correlated. In general, i.e., in the case of quantum side-information, the state $\rho_{ABE}$ could be completely arbitrary.

As we said, the devices will be used $n$ times in succession.
Let $\mb{V}=(V_1,\ldots, V_n)=(V_1^a,V_1^b\ldots, V^a_n,V^b_n)$ denote the sequence of inputs employed in $n$ such successive uses and let $P(\mb{v})$ denote the probability of a particular sequence $\mb{V}=\mb{v}$.  We assume that the choice of inputs is independent of the device behaviour, i.e., that the inputs $\mb{V}$, the pair of devices $AB$, and the adversary's information $E$ can initially be characterized by the $cqc$-state
\be \label{vxeinit}
\rho_{\mb{V}}\otimes\rho_{ABE}=\sum_{\mb{v},e} P(\mb{v})P(e) |\mb{v}\rangle\langle\mb{v}|\otimes \rho_{AB}^e\otimes |e\rangle\langle e|\,.
\ee
After the $n$ uses of the devices, one obtains a sequence $\mathbf{X}=(X^a_1,X^b_1,\ldots,X^a_n,X^b_n)$ of output pairs. The resulting situation, and the correlations between the inputs $\mb{V}$, outputs $\mb{X}$, and the adversary's information $E$, can then be characterized by the joint distribution
\be \label{P}
P(\mb{v}\mb{x}e)=P(\mb{v})P(e)P(\mb{x}|\mb{v},e)\,,
\ee 
where
\be\label{pxve}
P(\mb{x}|\mb{v},e)=\mr{tr}\left[\prod_{i=1}^n \left(U_{AB}\,M_A(x^a_i|v^a_i)\otimes  M_B(x^b_i|v^b_i)\right)\rho^e_{AB}\prod_{i=1}^n \left(M^\dagger_A(x^a_i|v^a_i)\otimes  M^\dagger_B(x^b_i|v^b_i)U^\dagger_{AB}\right)\right]
\ee
represents the response of the devices to given inputs $\mb{v}$ for a given value of the adversary's information $e$. 

In the following, we show how the level of Bell violation which is observed after $n$ repetitions of the experiment implies a bound on the min-entropy of the output string $\mb{X}$ conditioned on the input string $\mb{V}$ and the adversary's information $E$.  This bound depends only on the product assumption (\ref{prod}) characterizing the two devices, on the independence assumption (\ref{vxeinit}) between the choice of inputs and the state of the devices, and implicitly  on the condition (\ref{csi}) that the adversary's side-information is classical. Apart from these three assumptions, our results do not depend on any specific details of the device behaviour $(\rho_{ABE},\mathcal{M}_{AB},U_{AB})$.

\subsection{Bounding the min-entropy}
Suppose that the sequence of inputs $\mathbf{V}=(V^a_1,V^b_1,\ldots,V^a_n,V^b_n)$ is generated by choosing each pair of inputs $(V^a_i,V^b_i)$ independently with probability
$\mr{Pr}\left[V^a_i=v,V^b_i=w\right]=p_{vw}$, with $q=\min_{v,w} p_{vw}>0$.
Let $I$ be a Bell expression $I$ adapted to the input and output alphabet of the quantum devices. We then introduce the following Bell estimator
\be \bar I=\frac{1}{n}\sum_{i=1}^n I_i\,, \label{bari}\ee where
\be\label{Ii}
I_i=\sum_{xyvw} c_{xyvw}
\frac{\chi(X^a_i=x,X^b_i=y,V^a_i=v,V^b_i=w)}{p_{vw}}\,. \ee
Here,
$\chi(e)$ is the indicator function for the event $e$, that is,
$\chi(e)=1$ if the event $e$ is observed, $\chi(e)=0$ otherwise.
The series of coefficients $c_{xyvw}$ in (\ref{Ii}) define the Bell expression $I$. We assume that they satisfy $c=\max_{x,y,v,w}c_{xyvw}<\infty$. 

Let $\{J_m\,:\: 0\leq m\leq m_{\mr{max}}\}$ be a series of Bell violation thresholds, with $J_0$ corresponding to the local bound of the Bell expression and $J_{\mr{max}}=I_q$ to the maximum violation allowed by quantum theory.
We are going to put a bound on the min-entropy of the string $\mb{X}$ conditioned on the fact that the observed Bell average value $\bar I$ is comprised within some interval\footnote{This is the novel ingredient that fixes the issue in \cite{dirng}.} $J_m\leq \bar I< J_{m+1}$. We denote $P(m)$ the probability that the experiment returns a Bell average value comprised between $J_m\leq \bar I< J_{m+1}$ and $H_{\mr{min}}(\mb{X}|\mb{V}E,m)_P$ the min-entropy of $\mb{X}$ conditioned on $\mb{V}$ and $E$ given that a specific value $m$ has been obtained. The case $m=0$ corresponds to the situation where no substantial Bell violations is observed and no randomness is produced. 

\begin{theorem}\label{Theo1}
Suppose that the sequence of inputs $\mathbf{V}=(V^a_1,V^b_1,\ldots,V^a_n,V^b_n)$ is generated by choosing each pair of inputs $(V^a_i,V^b_i)$ independently with probability
$\mr{Pr}\left[V^a_i=v,V^b_i=w\right]=p_{vw}$, with $q=\min_{v,w} p_{vw}>0$. Let $\epsilon,\epsilon'>0$ be two arbitrary parameters.
Then for any device behaviour $(\rho_{ABE},\mathcal{M}_{AB},U)$, the resulting distribution $P=\{P(\mb{v}\mb{x}e)\}$ characterizing $n$ successive use of the devices is $\epsilon$-close to a distribution $Q$ such that
\begin{enumerate}
\item either $Q(m) \leq \epsilon'$,
\item or $H_{\mr{min}}(\mb{X}|\mb{V}E,m)_Q\geq nf(J_m-\mu)-\log_2\frac{1}{\epsilon'}$\,,
\end{enumerate}
where $f$ is a randomness bound associated to the Bell expression $I$ and 
\be
\mu=\left(\frac{c}{q}+I_q\right)\sqrt{\frac{2}{n}\ln\frac{1}{\epsilon}}\,.
\ee
\end{theorem}
This result tells us that the classical distribution $P$ characterizing the outputs $\mb{X}$ of the devices and their correlations with the inputs $\mb{V}$ and the adversary's information $E$ is essentially indistinguishable from a distribution $Q$ such that if the observed violation lies within the interval  $J_m\leq \bar I< J_{m+1}$ with non-negligible probability, then we have the guarantee that the outputs contain a certain amount of entropy, roughly given by $nf(J_m)$ up to epsilonic corrections (remark the term $-\log_2 1/\epsilon'$ in the bound on the min-entropy $H_{\mr{min}}(\mb{X}|\mb{V}E,m)_Q$ which was missing in \cite{dirng}).
Note that the fact that the trace distance cannot increase under classical-processing operations guarantees that any claim about the string $\mb{X}$ (or any subsequent use thereof) which is based on the properties of the distribution $Q$ will also hold for the distribution $P$ up to a correction $\epsilon$ (see subsection 3.4 for more details).

\paragraph{Proof of Theorem 1.}
In the following, we write $\mb{v}_i=(v^a_1,v^b_1\ldots,v^a_i,v^b_i)$ for the collection of input pairs up to round $i$, and similarly for $\mb{x}_i$. We denote $\mathbb{E}(I_i|\mb{x}_{i-1},\mb{v}_{i-1},e)$ the expectation of the random variable $I_i$ defined in (\ref{Ii}) conditioned on $(\mb{x}_{i-1},\mb{v}_{i-1},e)$, where the expectation is taken with respect to the probability distribution $P$. The following Lemma puts a bound on the probabilities $P(\mb{x}|\mb{v},e)$. 
\begin{lemma}\label{logp}
Let $G_\mu=\{(\mb{x},\mb{v},e)\mid \frac{1}{n}\sum_{i=1}^n
\mathbb{E}(I_i|\mb{x}_{i-1},\mb{v}_{i-1},e)\geq
\bar I(\mb{x},\mb{v})-\mu\}$, where $\mu\in\mathbb{R}$ is some real parameter.
Then for any $(\mb{x},\mb{v},e)\in G_\mu$,
\be \label{probfi}
P(\mb{x}|\mb{v},e)\leq g^n(\bar I(\mb{x},\mb{v})-\mu)\,.
\ee
\end{lemma}
\begin{proof} 
Using successively Bayes's rule and (\ref{pxve}), we can write
\be 
P(\mb{x}|\mb{v},e)=\prod_{i=1}^n P(x_i|v_i,\mb{x}_{i-1},\mb{v},e)
 =\prod_{i=1}^n P(x_i|v_i,\mb{x}_{i-1},\mb{v}_{i-1},e)\,.
\ee
The second equality simply expresses the fact that the outputs at round $i$ are determined only by the inputs at round $i$ and by the past inputs and outputs, but not by future inputs.
Note furthermore that we can write 
\begin{eqnarray}
P(x_i|v_i,\mb{x}_{i-1},\mb{v}_{i-1},e)&=&P(x^a_i,x^b_i|v^a_i,v^b_i,\mb{x}_{i-1},\mb{v}_{i-1},e)\nonumber\\
&=&\mr{tr}[M_A(x^a_i|v^b_i)\otimes M_B(x^b_i|v^b_i)\,\rho^{e,\mb{x}_{i-1},\mb{v}_{i-1}}_{AB}\,M^\dagger_A(x^a_i|v^b_i)\otimes M^\dagger_B(x^b_i|v^b_i)]\,,\nonumber \\ \label{ass2}
\end{eqnarray}
where $\rho^{e,\mb{x}_{i-1},\mb{v}_{i-1}}_{AB}$ denotes the state of the devices conditioned on previous inputs and outputs. Applying the randomness bound (\ref{prI}) to the probability distribution $P_{\mb{x}_{i-1},\mb{v}_{i-1},e}=\{P(x_i|v_i,\mb{x}_{i-1},\mb{v}_{i-1},e)\}$ implies that $P(x_i|v_i,\mb{x}_{i-1},\mb{v}_{i-1},e)\leq g(I[P_{\mb{x}_{i-1},\mb{v}_{i-1},e}])$. 
Using the fact that $P(v^a_i=v,v^b_i=w|\mb{x}_{i-1},\mb{v}_{i-1},e)=p_{vw}$ which follows from (\ref{vxeinit}) and the fact that each pair of inputs $(V^a_i,V^b_i)$ is generated independently with probability
$\mr{Pr}\left[V^a_i=v,V^b_i=w\right]=p_{vw}$, it is easily verified that $I[P_{\mb{x}_{i-1},\mb{v}_{i-1},e}]=\sum_{xyvw}c_{xyvw}P(xy|vw,\mb{x}_{i-1},\mb{v}_{i-1},e)=\mathbb{E}(I_i|\mb{x}_{i-1},\mb{v}_{i-1},e)$. We therefore have
\be 
P(\mb{x}|\mb{v},e) \leq \prod_{i=1}^n  g(\mathbb{E}(I_i|\mb{x}_{i-1},\mb{v}_{i-1},e))\leq g^n(\frac{1}{n}\sum_{i=1}^n\mathbb{E}(I_i|\mb{x}_{i-1},\mb{v}_{i-1},e))
\ee
where we used that $g$ is logarithmically concave in the second inequality. Using the definition of $G_\mu$ and the fact that $g$ is monotonically decreasing , we get (\ref{probfi}).
\end{proof}
\begin{lemma}\label{azuma}
For any $\epsilon>0$, let 
\be\label{mu}
\mu=\left(\frac{c}{q}+I_q\right)\sqrt{\frac{2}{n}\ln\frac{1}{\epsilon}}\,.
\ee
Then
\be 
\mr{Pr}[G_\mu]=\sum_{(\mb{x},\mb{v},e)\in G_\mu} P(\mb{x},\mb{v},e) \geq 1-\epsilon\,.
\ee
\end{lemma}
\begin{proof}
Consider the list of random variables ${Z}_0,\ldots,{Z}_n$, where ${Z}_0=0$ and
\be 
Z_k=\sum_{i=1}^k \left(I_k-\mathbb{E}(I_k|{W}_{k-1})\right)
\ee
for $k\geq 1$, where ${W}_{k-1}=(\mb{X}_{k-1},\mb{V}_{k-1},E)$ and $W_0=E$.
Since $|I_k|\leq c/q$ with $c< \infty$ and $q>0$, we have that $|{Z}_k|\leq 2kc/q< \infty$ is bounded for all $k$. Moreover, the differences $|Z_{k+1}-Z_k|$ are bounded by $|Z_{k+1}-Z_k|=|I_{k+1}-\mathbb{E}(I_{k+1}|{W}_k)|\leq |I_{k+1}|+|\mathbb{E}(I_{k+1}|{W}_k)|\leq c/q + I_q$, where we used (\ref{vxeinit}) and the fact that each pair of
inputs $(V^a_i,V^b_i)$ is generated independently with probability
$\mr{Pr}\left[V^a_i=v,V^b_i=w\right]=p_{vw}$. Finally, it is easily verified that $\mathbb{E}(Z_{k+1}|{W}_k)=Z_k$ for all $0\leq k\leq {n-1}$. The variables ${Z}_0,\ldots,{Z}_n$ thus form a martingale with respect to (the filtration induced by) ${W}_0,\ldots,{W}_{n-1}$. We can therefore apply  Azuma-Hoeffding inequality~\cite{GS}, which yields
\be
\mr{Pr}\left[Z_n-Z_0\geq n\mu\right]=\mr{Pr}\left[\frac{1}{n}\sum_{i=1}^n \mathbb{E}(I_i|{W}_{i-1})\leq \bar I-\mu\right]\leq \exp\left(\frac{-n\mu^2}{2(c/q+I_q)^2}\right)=\epsilon\,,
\ee
which gives the desired claim given the definition of $G_\mu$. 
\end{proof}
So far, we have (implicitly) considered the random variable sequence $\mb{X}$ as taking value in the output space $\mathcal{X}^{n}=\mathcal{X}\times\ldots\times\mathcal{X}$. We now formally extend its range and view it as an element of $\mathcal{X}^n\cup\perp$ (with $P(\mb{x}|\mb{v}e)=0$ if $\mb{x}=\perp$. We can interpret $\perp$ as an ``abort-output" produced by the devices implying that no violation has been obtained (i.e. $m=0$ if $\mb{x}=\perp$). 
\begin{lemma}\label{lemq}
There exists a probability distribution
$Q=\{Q(\mb{x},\mb{v},e)\}$ that is $\epsilon$-close to
$P$ satisfying
\be \label{logq}
Q(\mb{x}|\mb{v},e)\leq g^n(\bar I(\mb{x},\mb{v})-\mu)\,.
\ee
for all $(\mb{x},\mb{v},e)$ such that $\mb{x}\neq \perp$, with $\mu$ given by (\ref{mu}).
\end{lemma}
\begin{proof}
Define $Q$ as $Q(\mb{x},\mb{v},e)=P(\mb{v})P(e)Q(\mb{x}|\mb{v},e)$, where $Q(\mb{x}|\mb{v},e)=P(\mb{x}|\mb{v},e)$ if $(\mb{x},\mb{v},e)\in G_\mu$, $Q(\mb{x}|\mb{v},e)=0$ if $\mb{x}\neq \perp$ and $(\mb{x},\mb{v},e)\notin G_\mu$, and $Q(\perp|\mb{v},e)=1-\sum_{\mb{x}\neq G_\mu}P(\mb{x}|\mb{v},e)$. By Lemma~\ref{logp}, the distribution $Q$ satisfies (\ref{logq}) for all $(\mb{x},\mb{v},e)$ such that $\mb{x}\neq \perp$. Application of Lemma~\ref{azuma} gives $d(P,Q)=\frac{1}{2}\sum_{\mb{x},\mb{v},e}|P(\mb{x},\mb{v},e)-Q(\mb{x},\mb{v},e)|=\frac{1}{2}\sum_{\mb{v},e}P(\mb{v},e)\sum_{\mb{x}}|P(\mb{x}|\mb{v},e)-Q(\mb{x}|\mb{v},e)|= \frac{1}{2}(\sum_{\mb{x},\mb{v},e\notin G_\mu}P(\mb{x},\mb{v},e)+1-\sum_{\mb{x},\mb{v},e\in G_\mu}P(\mb{x},\mb{v},e))\leq\epsilon$.
\end{proof}
Let $Q(m)$ be the probability (according to the distribution $Q$) that $J_m\leq \bar I< J_{m+1}$. Let $Q(\mb{x},\mb{v},e|m)$ denote the distribution of $\mb{X},\mb{V},E$ conditioned on a particular value of $m$ and let
\be\label{hminq}
H_\mr{min}(\mb{X}|\mb{V}E,m)_Q=-\log_2\sum_{\mb{v},e}
Q(\mb{v},e|m)\max_{\mb{x}}Q(\mb{x}|\mb{v},e,m)
\ee
be the min-entropy of the raw string $\mb{X}$ conditioned on ($\mb{V},E$) for a given $m$. Let $K_m=\{\mb{x}\,|\,\mb{x}\neq \perp \text{ and }J_{m}\leq \bar I(\mb{x},\mb{v})<J_{m+1}\}$. By Lemma~\ref{lemq} and the fact that the $g$ is monotically decreasing, we have
\begin{eqnarray}
\max_{\mb{x}}Q(\mb{x}|\mb{v},e,m)&=&\frac{1}{Q(m|\mb{v},e)}\max_{\mb{x}\in K_m}Q(\mb{x}|\mb{v},e)\\
&\leq& \frac{g^n(J_m-\mu)}{Q(m|\mb{v},e)}\,.\label{maxm}
\end{eqnarray}
Inserting this back in (\ref{hminq}) gives
\begin{eqnarray} H_\mr{min}(\mb{X}|\mb{V}E,m)_Q&\geq
&-\log_2\sum_{\mb{v},e}
\frac{Q(\mb{v},e|m)}{Q(m|\mb{v},e)}g^n(J_m-\mu)\\
&=&-\log_2\sum_{\mb{v},e}
\frac{Q(\mb{v},e)}{Q(m)}g^n(J_m-\mu)\\
&=&nf(J_m-\mu)-\log_2\frac{1}{Q(m)}\,,
\end{eqnarray}	
where we remind that $f=-\log_2 g$. This immediately implies Theorem~1.

\subsection{Application to DIRE protocols}
Theorem~1 can directly be applied to prove the security of various DIRE protocols. Formally, a randomness expansion protocol is a protocol that, starting from a $d$-bit uniform random seed $S$, generates a $m$-bit string $R$ that is close to uniformly random and uncorrelated from any potential
adversary.  The length $m$ of the output string is
variable and determined during the run of the protocol. The protocol may also abort, in which case we set $m=0$ and $R=\emptyset$. We can assume that $m$ is made public at the end of the protocol.

The protocol will involve the use of Bell-violating devices and some classical processing on the outputs of the devices. For example a straightforward protocol directly based on the simple Bell experiment described so far is described below. But one could also consider more complicated protocols involving multiple pairs of Bell-violating devices, where this simple primitive is repeated or concatenated, see \cite{smv2,fgsv2}.

\begin{enumerate}
\item[\textbf{1}.]\textbf{Input generation:}
Alice generates a sequence of input pairs
$\mathbf{V}=(V^a_1,V^b_1,\ldots,V^a_n,V^b_n)$ according to the (non-uniform) distribution specified in the statement of Theorem~1. This can be achieved starting from a uniform random seed $S_{\mr{inp}}$ with a small error $\epsilon_{\mr{inp}}$ and small entropy loss (see \cite{ct,ky} and the Appendix). 
\item[\textbf{2}.]\textbf{Use of the devices:}
She introduces inputs $V^a_i$ and $V^b_i$ in the two
devices and obtains outputs $X^a_i$ and $X^b_i$. This step is
repeated $n$ times, resulting in the sequence of output pairs
$\mathbf{X}=(X^a_1,X^b_1,\ldots,X^a_n,X^b_n)$.
\item[\textbf{3}.]\textbf{Estimation of the Bell expression:}
Alice computes the average Bell expression (\ref{bari}) and 
determines the value of $m$ such that $J_m\leq \bar I< J_{m+1}$. If $m=0$, she aborts.
\item[\textbf{4}.]\textbf{Randomness extraction:} Using a random seed $S_\mr{ext}$, Alice applies a $(m,k_m,\epsilon_{\mr{ext}})$-randomness extractor to the raw
string ${\mb{X}}$ with $k_m=nf(J_m-\mu)-\log_2 m_{\mr{max}}-\log_2\frac{1}{\epsilon'}$  
and obtains a string $R=\mr{Ext}(\mb{X},S_\mr{ext})$, which represents the output of the protocol. 
We can assume that $m$, $\mb{V}$, and $S_{\mr{ext}}$ are made public.
\end{enumerate}
In the above description, we have of course implicitly assumed that the thresholds $J_m$, the parameter $\epsilon$ (which determine $\mu$), $\epsilon'$, and $\epsilon_{\mr{ext}}$ are chosen in such a way that they define a proper $(m,k_m,\epsilon_{\mr{ext}})$-randomness extractor for all values of $m=1,\ldots,m_\mr{max}$.

Let $F=(\mb{V},S_{\mr{ext}},E)$ denote the final side-information of the adversary. Following the definition of security in the context of quantum key distribution outlined in \cite{ucqkd,charles}, we say that a protocol such as the one just presented is secure if, for any device behaviour and any $m$, the
output $R$ is uniformly random and independent from $F$. This means that the distribution $P^{\mr{perf}}_{RFM}$ characterizing the output $R$, the side-information $F$ and the final length $M$ of a perfectly secure
protocol has the form
\be \label{perf}
P_{RFM}^{\mr{perf}}(rfm)=P_M(m)\times P_{RF|M}(rf|m)\quad \text{with } P_{RF|M}(rf|m)=U_m(r)\times P_{f|M}(f|m)
\ee
where $U_m$ is the uniform distribution on $\{0,1\}^m$.
A real DIRE protocol is said to be
$\epsilon_{\mr{sec}}$\emph{-secure} if it is
$\epsilon_{\mr{sec}}$-indistinguishable from a secure protocol,
that is, if for any device behaviour, the joint distribution $P_{RFM}$ satisfies
\be \label{real}
d(P_{RFM},P_{RFM}^{\mr{perfect}})\leq 
\epsilon_{\mr{sec}} \ee
for some distribution 
$P_{RFM}^{\mr{perfect}}$ of the form (\ref{perf}). In
particular, a DIRE protocol is $\epsilon_{\mr{sec}}$-secure
if, for any device behaviour, it outputs $m$-bit strings
that are $\delta_m$-random with respect to $E$ with \be
\label{cond} \sum_{m=1}^{m_{\mr{max}}} P_M(m)\, \delta_m\leq
\epsilon_{\mr{sec}}\,,\ee
where $m_{\mr{max}}$ denotes the
maximal output length. 

To show that the protocol defined above is secure according to this definition, suppose that at the end of Step 2, after the $n$ uses of the devices, the correlations between outputs $\mb{X}$, inputs $\mb{V}$, and the adversary's prior information $E$ are characterized by the probability distribution $Q_{\mb{X}\mb{V}E}$ defined in the statement of Theorem~1. Then it is easy to show that the distribution $Q_{RFM}=Q_{G(\mb{X},\mb{V},S_\mr{ext})FM}$ characterizing the final output of the protocol (where $G$ is the classical processing describing the steps performed after the $n$ uses of the devices) is $(\epsilon'+\epsilon_\mr{ext})$-close to a perfectly secure distribution $\tilde Q_{RFM}$. Indeed, let $M_<$ be the values of $m$ such that $Q(m)\leq {\epsilon'}/{m_{\max}}$ and $M_>$ those for which $Q(m)\geq {\epsilon'}/{m_{\max}}$. For all $m\in M_>$, the min-entropy $H_\mr{min}(\mb{X}|\mb{V}E,m)_Q$ can thus be bounded by
\be \label{hminm}
H_\mr{min}(\mb{X}|\mb{V}E,m)_Q\geq nf(J_m-\mu)-\log_2 m_{\mr{max}}-\log_2\frac{1}{\epsilon'}\,.
\ee 
Applying a $(m,k_m,\epsilon_{\mr{ext}})$-randomness extractor to the string
$\mb{X}$ with $k_m$ given by the right-hand side of (\ref{hminm}) therefore yields a string that is $\delta_m$-close to a random string, with $\delta_m\leq \epsilon_{\mr{ext}}$ for $m\in M_>$ and $\delta_m\leq 1$ for $m\in M_<$. On average, we thus have
\be 
\sum_m Q(m)\delta_m\leq \sum_{m\in M_<}Q(m) +\sum_{m\in M>} Q(m)\epsilon_{\mr{ext}}\leq \sum_{m\in M_<} \frac{\epsilon'}{m_\mr{max}} +\sum_{m\in M_>} Q(m)\epsilon_{\mr{ext}}\leq \epsilon'+\epsilon_{\mr{ext}}\,.
\ee
Since the actual distribution $P_{\mb{X}\mb{V}E}$ characterizing the output of the device is $\epsilon$-close to 
$Q_{\mb{X}\mb{V}E}$, it directly follows that it provides an  $(\epsilon+\epsilon'+\epsilon_{\mr{ext}})$-secure realization of the protocol. Indeed, by the triangle inequality, and the fact that classical processing can 
only reduce the trace distance, we find $d(P_{{RFM}},\tilde Q_{RFM})\leq d(P_{{RFM}},Q_{{RFM}})+d(Q_{{RFM}},\tilde Q_{{RFM}})\leq \epsilon+\epsilon'+\epsilon_{\mr{ext}}$. By the same argument, the protocol is $(\epsilon_{\mr{inp}}+\epsilon+\epsilon'+\epsilon_{\mr{ext}})$-secure when errors inherent to the input generation are taken into account (see the Appendix for an analysis of the errors introduced at this stage).  

More generally, the security (in the context of classical side-information) of more complex protocols, such as those considered in \cite{smv2,fgsv2}, where outputs of one pair of devices are used as inputs for another pair of devices can directly be proven from Theorem~1 and by keeping track of the error propagation.

\subsection*{Efficiency}
While the protocol presented above produces new randomness, it also uses a source of initial randomness $S=(S_\mr{inp},S_{\mr{ext}})$ to generate the inputs $\mb{V}$ and perform the final randomness extraction. As a straightforward generalization of condition (\ref{vxeinit}), the  security of the protocol requires this initial seed to be uniform and independent from the initial state of the devices, i.e.,
\be \label{sindep}
\rho_{SABE}=\omega_S\otimes\rho_{ABE}
\ee
where $\omega_S$ denotes the uniform distribution on $S$. 

This condition is obviously satisfied if $S$ represents the output of a genuine, cryptographically-secure random number generator. Of course, a device-independent randomness expansion protocol is useful only if it produces more randomness at its output than is consumed at its input.  It is shown in \cite{dirng} how the protocol that we have presented above can achieve quadratic expansion by choosing appropriately the probabilities $p_{vw}$ characterizing the input distribution. It can also be used as a primitive in more elaborate protocols where the output of one pair of devices are repeatedly used as input for another pair of devices. Such protocols can achieve exponential expansion, see \cite{smv2} and particularly Section 5 of \cite{fgsv2} for quantitative details (note that the application of our results, valid against classical-side information, to such concatenated protocols require not only that different pairs of devices be unentangled to the adversary to start with, but also between themselves. This assumption is again very reasonable in a trusted-provider situation). 

Note, however, that to generate private randomness, a device-independent protocol does not necessarily need to consume any cryptographically-secure randomness to start with. 
 Indeed, since we assumed in the security analysis that $S$ was made public, the seed $S$ does not need to be random with respect to the adversary, provided that condition (\ref{sindep}) is satisfied, i.e., provided that the adversary cannot exploit any prior knowledge about $S$ to influence the behaviour of the devices. If this is the case, which may be reasonable to assume in a trusted provider situation\footnote{Note though that even in a trusted provider situation, the condition (\ref{sindep}) may fail, if the adversary can modify the behaviour of the devices by controlling external parameters like, e.g., the power supply of the devices.}, the output of the protocol will nevertheless represent randomness that is private with respect to the adversary.

\section*{Acknowledgements.} We thank Ll. Masanes for pointing out an error in a previous version of this paper and Serge Fehr and Chrisitian Schaffner for useful discussions.  This work was
supported by the European EU QCS project, the CHIST-ERA DIQIP project, the Interuniversity Attraction
Poles Photonics@be Programme (Belgian Science Policy), the
Brussels-Capital Region through a BB2B Grant.

\section*{Appendix}

Here we prove that one can use a uniform distribution to efficiently sample with exponentially small error from an i.i.d. non-uniform distribution, see also \cite{ct,ky}.

\begin{theorem}
Consider the finite alphabet $K={a_{1},...,a_{|K|}}$. Let $Q$ be a probability distribution on $K$ with  $\min_{a}Q(a)=n^{-\gamma}.$ 
Let $a^{n}=a_{1},a_{2},...,a_{n}\in K^{n}$ be drawn i.i.d. according
to $Q.$ We denote $Q^{n}$ the corresponding probability distribution
on $K^{n}$. Suppose that $x\in\left\{ 0,1\right\} ^{m}$ is drawn from the uniform
distribution $\omega$ on $m$ bits. Then, for any $0\leq\gamma<1/3$, one can construct 
a function $f:\left\{ 0,1\right\} ^{m}\to K^{n}$
such that the induced probability distribution on $K^{n}$ given by
$P(a^{n})=\omega\left(f^{-1}(a^{n})\right)$ is $\epsilon$ close
to $Q^{n}$, i.e., $d(P,Q^{n})=\frac{1}{2}\sum_{a^n}|P(a^n)-P'(a^n)|\leq\epsilon$ with $m\geq nH(Q)+o(nH(Q))$ and $\epsilon\leq3\exp\left[-2n^{1-3\gamma}\right]$, where $H(Q)=-\sum_a Q(a)\ln Q(a)$ is the Shannon entropy of $Q$.
\end{theorem}
\begin{proof}
The proof follows from Lemmas~\ref{iid1}, \ref{iid2}, \ref{iid3} below. Lemma~\ref{iid1} shows that there
is a probable subset of $K^{n}$ which occurs with high probability,
and Lemma~\ref{iid2} computes the size of this probable subset. In Lemmas~\ref{iid1}
and \ref{iid2}, we take parameter $\alpha=n^{1/2-\gamma}$. With this choice,
from Lemma \ref{iid1}, the error one makes is $\leq2\exp\left[-2n^{1-3\gamma}\right]$
and from Lemma \ref{iid2} the size of the probable subset is $\leq2^{nH(Q)+O(n^{1-\gamma})}$.
Finally Lemma \ref{iid3} tells us how one can sample efficiently from a distribution
of known size. We take the error parameter in Lemma \ref{iid3} to be $\exp\left[-2n^{1-3\gamma}\right]$
(i.e. the same as in Lemma \ref{iid1}). The additional size penalty is negligible
compared to the one coming from Lemma \ref{iid2}. This proves the result.
\end{proof}

\emph{Counting Typical sequences.}
Consider the alphabet $K={a_{1},...,a_{|K|}}$. If $a^{n}=a_{1},a_{2},...,a_{n}\in K^{n}$
is a word of length $n$ we denote by $N(a|x)$= number of occurences
of $a\in K$ in word $a^{n}$ (this is known as the type of the sequence).
Let $Q$ be a probability distribution on $K$. Let $a^{n}=a_{1},a_{2},...,a_{n}\in K^{n}$
be drawn i.i.d. according to $Q.$ We denote $Q^{n}$ the corresponding
probability distribution on $K^{n}$.

For any $\alpha>0$ define the set: 
\[
T_{Q\alpha}^{n}=\left\{ x\in K^{n}:\forall a\in K\ |N(a|x)-nQ(a)|\leq\alpha\sqrt{n}\sqrt{Q(a)}\right\} 
\]

\begin{lemma}\label{iid1}
$Q^{n}\left(T_{Q\alpha}^{n}\right)\geq1-2|K|\exp\left[-2\alpha^{2}\min_{a}Q(a)\right]$.
\end{lemma}
\begin{proof}
$T_{Q\alpha}^{n}$ is the intersection of $|K|$ events, namely
that for each $a\in K$ the mean of the i.i.d Bernouilli variables
$y_{i}$, defined by $y_{i}=1$ iff. $a_{i}=a$ and $y_{i}=0$ iff.
$a_{i}\neq a$, has deviation from its expected value $Q(a)$ by at
most $\alpha\sqrt{n}\sqrt{Q(a)}$. By the Hoeffding bound, each of
these events has probability $\geq1-2\exp\left[-2\alpha^{2}Q(a)\right]$.
Hence the intersection of the events has probability $\geq1-2|K|\exp\left[-2\alpha^{2}\min_{a}Q(a)\right]$.
\end{proof}
\begin{lemma}\label{iid2}
$|T_{Q\alpha}^{n}|\leq2^{nH(Q)+2\frac{\log_{2}e}{e}|K|\alpha\sqrt{n}}$.
\end{lemma}
\begin{proof}
Consider $x\in T_{Q\alpha}^{n}$. Then $Q(x)=\prod_{a\in K}Q(a)^{N(a|x)}$.
Hence 
\begin{eqnarray*}
|-\log_{2}Q(x)-nH(Q)| & = & |\sum_{a\in K}-N(a|x)\log_{2}Q(a)-nH(Q)|\\
 & \leq & \sum_{a\in K}-\log_{2}Q(a)\ |N(a|x)-nQ(a)|\\
 & \leq & \sum_{a\in K}-\log_{2}Q(a)\ \alpha\sqrt{Q(a)}\sqrt{n}\\
 & = & 2\alpha\sqrt{n}\sum_{a\in K}-\log_{2}\sqrt{Q(a)}\ \sqrt{Q(a)}\\
 & \leq & 2\alpha\sqrt{n}\frac{\log_{2}e}{e}|K|
\end{eqnarray*}

Therefore $Q(x)\geq2^{-nH(Q)-2\frac{\log_{2}e}{e}|K|\alpha\sqrt{n}}$,
and $1\geq\sum_{x\in T_{Q\alpha}^{n}}Q(x)\geq|T_{Q\alpha}^{n}|2^{-nH(Q)-2\frac{\log_{2}e}{e}|K|\alpha\sqrt{n}}$
which proves the result.
\end{proof}

\emph{Sampling from arbitrary distributions.}
Suppose that $x\in\left\{ 0,1\right\} ^{m}$ is drawn from the uniform
distribution $\omega$.

Consider the probability distribution $P(z)$ on $z\in\left\{ 0,1\right\} ^{k}$.
We want to use $x$ to sample with high precision from $P(z)$. That
is, we define a function $f:\left\{ 0,1\right\} ^{m}\to\left\{ 0,1\right\} ^{k}:x\to f(x)$
such that the induced probability distribution $P'(z)=\omega(f^{-1}(x))$
is close to $P(z)$, as measured by the trace distance $d(P,P')=\frac{1}{2}\sum_{z}|P(z)-P'(z)|$.
We have:

\begin{lemma}\label{iid3}
For any $\epsilon>0$, if $m\geq k+\log_{2}\frac{1}{\epsilon}$
we can construct a function $f$ such that $d(P,P')\leq\epsilon$.
\end{lemma}
\begin{proof}
We view any $x\in\left\{ 0,1\right\} ^{m}$ as a number in
$[0,1]$ written in binary: $x=\sum_{i=1}^{m}x_{i}2^{-i}$.

We define $P'(z)\in\left\{ 0,1\right\} ^{m}$ as the largest binary
number smaller than $P(z)$. Therefore $0\leq P(z)-P'(z)\leq2^{-m}$.
We have $1-\sum_{z}P'(z)=\sum_{z}P(z)-P'(z)\leq2^{-(m-k)}$. To have
a normalised distribution we define an additional outcome $\perp$
with $P'(\perp)=1-\sum_{z}P'(z)$ . Using $x\in\left\{ 0,1\right\} ^{m}$
drawn from the uniform distribution $\omega$, we can therefore sample
from $P'(z)$ thus defined with $d(P,P')=\frac{1}{2}\sum_{z}|P(z)-P'(z)|+\frac{1}{2}P'(\perp)\leq2^{-(n-k)}$.
(The function $f$ can be explicitly defined through $f^{-1}(z)=\left\{ x:\sum_{z'=0}^{z}P'(z')\leq x\leq\sum_{z'=0}^{z+2^{-k}}P'(z')\right\} $).
\end{proof}

\end{document}